\documentclass[aps,groupedaddress,superscriptaddress,amssymb,amsmath,amsbsy,amsfonts,twocolumn,pra]{revtex4}

\usepackage{amssymb,amsmath,amsthm,color,graphicx,times,graphicx}
\usepackage{hyperref}
\usepackage{enumerate}
\usepackage{subfigure}
\usepackage{dsfont}
\usepackage{times}
\usepackage{txfonts}
\usepackage{bbold}

\usepackage{extarrows}

\newcommand*{\cl}[1]{{\mathcal{#1}}}
\newcommand*{\bb}[1]{{\mathbb{#1}}}
\newcommand{\ket}[1]{\left|#1\right>}
\newcommand{\bra}[1]{\left<#1\right|}

\newcommand{\proj}[2]{| #1 \rangle\!\langle #2 |}
\newcommand*{\tn}[1]{{\textnormal{#1}}}
\newcommand*{\1}{{\mathbb{1}}}
\newcommand{\T}{\mbox{$\textnormal{Tr}$}}

\theoremstyle{plain}

\newtheorem{theorem}{Theorem}
\newtheorem{lemma}[theorem]{Lemma}
\newtheorem{corollary}[theorem]{Corollary}

\theoremstyle{definition}

\date{\today}

\begin{document}

\title{Generalization of port-based teleportation and controlled teleportation capability}

\author{Kabgyun Jeong}
\email{kgjeong6@snu.ac.kr}
\affiliation{Research Institute of Mathematics, Seoul National University, Seoul 08826, Korea}
\affiliation{School of Computational Sciences, Korea Institute for Advanced Study, Seoul 02455, Korea}
\author{Jaewan Kim}
\email{jaewan@kias.re.kr}
\affiliation{School of Computational Sciences, Korea Institute for Advanced Study, Seoul 02455, Korea}
\author{Soojoon Lee}
\email{level@khu.ac.kr}
\affiliation{Department of Mathematics and Research Institute for Basic Sciences, Kyung Hee University, Seoul 02447, Korea}
\affiliation{School of Computational Sciences, Korea Institute for Advanced Study, Seoul 02455, Korea}

\date{\today}

\begin{abstract}
As a variant of the original quantum teleportation, port-based teleportation has been proposed, 
and its various kinds of useful applications in quantum information processing have been explored. 
Two users in the port-based teleportation initially share an arbitrary pure state, 
which can be represented by applying one user's local operation to a bipartite maximally entangled state.
If the maximally entangled state is a $2M$-qudit state, 
then it can be expressed as $M$ copies of a two-qudit maximally entangled state, 
where $M$ is the number of ports. 
We here consider a generalization of the original port-based teleportation 
obtained from employing copies of an arbitrary bipartite (mixed) resource 
instead of copies of a pure maximally entangled one.
By means of the generalization, 
we construct a concept of controlled port-based teleportation by combining controlled teleportation with port-based teleportation,
and analyze its performance in terms of several meaningful quantities 
such as the teleportation fidelity, the entanglement fidelity, and the fully entangled fraction. 
In addition, we present quantities called the control power 
and the minimal control power for the controlled version on a given tripartite quantum state.
\end{abstract}

\maketitle

\section{Introduction}
Quantum teleportation, proposed by Bennett \emph{et al}.~\cite{BBC+93}, 
is a fundamental and innovative way to transmit an unknown quantum information from a sender to a remote receiver 
by exploiting a prior distributed entanglement~\cite{NC00,W13}. 
According to its potential applications, 
this scheme has been studied by using various methods in experimental regimes 
as well as in theoretical ways~\cite{BPM+97,FSB+98,KKS01,BBM+98,R+17}.

As a variant of quantum teleportation, 
port-based teleportation (PBT) has been suggested~\cite{IH08,IH09,WB16,SSMH17,CLM+18,MSSH18,PBP19}. 
While the standard teleportation requires a receiver's recovering operation at the end of the protocol, 
the PBT scheme does not need such a correction operation, 
and only requires the receiver to choose a \emph{port} depending on the classical information 
related to the sender's measurement outcome. 
It has been shown from the features of the PBT scheme 
that PBT can have several kinds of applications in quantum information processing 
such as a universal programmable quantum processor~\cite{IH08}, instantaneous non-local quantum computation~\cite{BK11}, 
quantum-channel discrimination~\cite{PLLP19}, and quantum telecloning protocols~\cite{MJPV99,PB11}. 

In the PBT protocol, 
two users initially prepare an arbitrary $2M$-qudit pure state shared between them.
The state can be decomposed as an arbitrary operation on one user's system and a pure maximally entangled state, 
which can be represented as $M$ copies of a two-qudit maximally entangled state. 
In this paper, we first employ copies of an arbitrary bipartite (mixed) state as the generalized resource 
instead of copies of the maximally entangled state in the PBT protocol,  
and analyze the generalized PBT.

We can also find another kind of the teleportation scheme over any tripartite pure states, 
called controlled teleportation (CT), 
which is a modification of the splitting and reconstruction of the pre-shared quantum information~\cite{KB98,LJK05,LJK07}. 
In the CT scheme, a controller's assistance through a local measurement 
can improve the teleportation fidelity between the sender and the receiver, 
and hence it is natural to take into account the controller's power in the teleportation procedure. 
Recently, the control power has been rigorously investigated in the perfect CT scheme, including higher dimensional cases~\cite{LG14,LG15} 
as well as the general CT ones~\cite{JKL16}.

By combining the two different kinds of teleportation above, 
we here suggest a concept called controlled port-based teleportation (CPBT), 
and analyze its performance by manipulating our generalized PBT. 

This paper is organized as follows. 
In Sec.~\ref{pre}, we describe the definitions of several meaningful quantities which we deal with here,
such as the teleportation fidelity, 
the entanglement fidelity, and the fully entangled fraction, 
illustrating the relations between the teleportation fidelity and the other two fidelities. 
Furthermore, we briefly introduce the main idea of the PBT (see Sec.~\ref{PBT}). 
In Sec.~\ref{main}, we provide a generalization of the original PBT, 
and derive the results related to the performance of our generalized PBT.
In Sec.~\ref{CPBT}, we provide a concept of 
the quantities for the teleportation's capabilities, 
and analyze its properties. 
Moreover, we study the controller's control power (also minimal control power) for the CPBT on a given tripartite quantum state. 
Finally, discussions and remarks are offered in Sec.~\ref{conclusion}, 
and some open questions are raised for future work.

\section{Preliminaries} \label{pre}
\subsection{Teleportation fidelity, entanglement fidelity, and fully entangled fraction} \label{definitions}
We briefly review the mathematical definitions and the relation 
between the teleportation fidelity, the entanglement fidelity, and the fully entangled fraction~\cite{HHH99,BHHH00}. 
First, let $\Lambda_{\varrho}$ be the standard teleportation channel over a bipartite quantum state $\varrho$; 
then the teleportation fidelity is naturally given by 
\begin{equation*}
\bb{F}_{\tn{T}}(\Lambda_\varrho)=\int\bra{\psi}\Lambda_\varrho(\psi)\ket{\psi}\tn{d}\psi,
\end{equation*}
where the integral is performed with respect to the uniform distribution $\tn{d}\psi$ 
over all $d$-dimensional pure states 
$\psi:=\proj{\psi}{\psi}$.
The entanglement fidelity is defined as 
\begin{equation} \label{eq:entfid}
F(\Lambda_{\varrho})=\T\Phi^+\left[(\Lambda_\varrho\otimes\1)\Phi^+\right],
\end{equation}
where $\Phi^+=\frac{1}{d}\sum_{i,j=0}^{d-1}\proj{ii}{jj}$ is a maximally entangled state with Schmidt rank $d$, 
and the fully entangled fraction of $\varrho$ is defined by
\begin{equation}\label{eq:fef}
f(\varrho)=\max_{\ket{e}}\bra{e}\varrho\ket{e},
\end{equation}
where the maximum is taken over all maximally entangled states $\ket{e}$ with Schmidt rank $d$. 
It is known that $F(\Lambda_\varrho)= f(\varrho)$ if $\Phi^+$ is equal to a pure state $\ket{e}\bra{e}$
which attains the maximum in Eq.~(\ref{eq:fef}), that is, $f(\varrho)=\bra{e}\varrho\ket{e}=\T \Phi^+\varrho=F(\Lambda_\varrho)$~\cite{HHH99}. 
Without loss of generality, we may assume that 
the entanglement fidelity $F(\Lambda_{\varrho})$ is equivalent to the fully entangled fraction $f(\varrho)$
by taking $\Phi^+$ in Eq.~(\ref{eq:entfid}) as $\ket{e}\bra{e}$ satisfying $f(\varrho)=\bra{e}\varrho\ket{e}$.

We now observe some useful relations between the fidelities. 
The teleportation fidelity and the entanglement fidelity over $\Lambda_\varrho$ take a universal relation~\cite{HHH99,BHHH00} in the form of
\begin{equation} \label{eq:rel}
\bb{F}_{\tn{T}}(\Lambda_\varrho)=\frac{dF(\Lambda_\varrho)+1}{d+1}.
\end{equation}
We remark that $\bb{F}_{\tn{T}}(\Lambda_\varrho)>\frac{2}{d+1}$ [or $F(\Lambda_\varrho)>\frac{1}{d}$] 
if and only if $\varrho$ is said to be {\em meaningful} for the teleportation, 
since it was shown that 
the classical teleportation can have a fidelity of at most $\bb{F}_{\tn{T}}(\Lambda_\varrho)=\frac{2}{d+1}$ 
[or $F(\Lambda_\varrho)=\frac{1}{d}$]~\cite{BHHH00}.

\subsection{Port-based teleportation} \label{PBT}

\begin{figure}[t!]
\includegraphics[width=\linewidth]{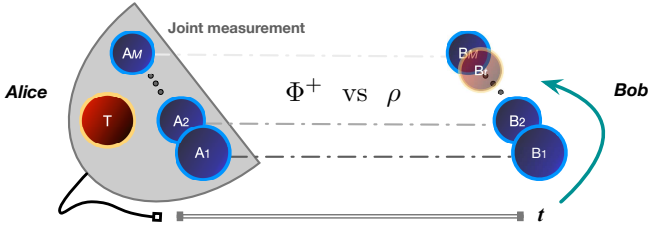}
\caption{A schematic diagram for the PBT. 
The standard PBT protocol is about the PBT on copies of the maximally entangled state $\Phi^+$, 
but we consider the PBT protocol on copies of any arbitrary bipartite state $\rho$. 
As in the original PBT protocol, 
Alice's joint measurement outcome $t$ is transmitted to Bob through the classical channel (double line),  
and Bob only selects the port $t$ to teleport the state $\psi_T$.}
\label{Fig1}
\end{figure}

Let $\psi_T:=\proj{\psi}{\psi}_T$ be a $d$-dimensional unknown pure state, 
which Alice wants to teleport to Bob. 
Let $\mathbf{A}=\{A_1,A_2,\ldots,A_M\}$ and $\mathbf{B}=\{B_1,B_2,\ldots,B_M\}$ 
denote the Alice's and Bob's total systems for $M$ ports, respectively. 
To begin with, we assume that the sender Alice and the receiver Bob share a $2M$-qudit pure state of the form 
\begin{align}
\varphi_{\mathbf{AB}}
&=(O_{\mathbf{A}}\otimes \1_{\mathbf{B}})\Phi_{A_1B_1}^+\otimes\cdots\otimes\Phi_{A_MB_M}^+(O_{\mathbf{A}}^\dag\otimes \1_{\mathbf{B}}),
\label{eq:varphi_AB}
\end{align}
where Alice's operations $O_{\mathbf{A}}$ satisfying $\T[O_{\mathbf{A}}O_{\mathbf{A}}^\dag]=d^M$. 
Let $\bar{A}_t:=\mathbf{A}\setminus \{A_t\}$ and $\bar{B}_t:=\mathbf{B}\setminus \{B_t\}$; 
then we can represent the PBT channel $\tilde{\Lambda}_{\varphi}$ over the state $\varphi_{\mathbf{AB}}$ as
\begin{align*}
\tilde{\Lambda}_{\varphi}(\psi_T)
&=\sum_{t=1}^M\left[\T_{\mathbf{A}\bar{B}_tT}\sqrt{\Pi_t^{(\mathbf{A}T)}}(\varphi_{\mathbf{AB}}\otimes\psi_T)
\sqrt{\Pi_t^{(\mathbf{A}T)}}^\dag\right]_{B_t\to B} 
\nonumber\\
&=\sum_{t=1}^M \T_{\mathbf{A}T}\Pi_t^{(\mathbf{A}T)}\left[(O_{\mathbf{A}}\otimes\1_B)\sigma_{\mathbf{A}B}^{(t)}(O_{\mathbf{A}}^\dag\otimes\1_B)\otimes\psi_T\right],
\end{align*}
where the positive operator-valued measurement elements are described by 
$\{\Pi_{i}^{(\mathbf{A}T)}\}_{i=1}^M$ such that $\sum_{i=1}^M\Pi_i^{(\mathbf{A}T)}=\1_{\mathbf{A}T}$, 
and 
\begin{align*}
\sigma_{\mathbf{A}B}^{(t)}&=\left[\T_{\bar{B}_t}\left(\Phi_{A_1B_1}^+\otimes\cdots\otimes\Phi_{A_MB_M}^+\right)\right]_{B_t\to B} \\
&=\frac{1}{d^{M-1}}\Phi_{A_tB}^+\otimes\1_{\bar{A}_t}.
\end{align*}

By exploiting Eq.~(\ref{eq:entfid}), we can obtain the entanglement fidelity for the channel $\tilde{\Lambda}_{\varphi}$~\cite{IH08,IH09} as
\begin{align*}
F(\tilde{\Lambda}_{\varphi})
&=\T\Phi_{BD}^+\left[(\tilde{\Lambda}_{\varphi}\otimes\1_D)\Phi_{TD}^+\right] \\
&=\frac{1}{d^2}\sum_{t=1}^M\T\Pi_t^{(\mathbf{A}B)}\left[(O_{\mathbf{A}}\otimes\1_B)\sigma_{\mathbf{A}B}^{(t)}(O_{\mathbf{A}}^\dag\otimes\1_B)\right].
\end{align*}

It was shown that if the PBT protocol is the deterministic standard one, 
that is, $O$ is the identity operator and $\{\Pi_{i}\}_{i=1}^M$ is the pretty good measurement,
then for sufficiently large $M\gg0$ and any $\varepsilon>0$, 
the entanglement fidelity $F(M)$ for the standard PBT protocol is given by~\cite{IH08,IH09,CLM+18}
\begin{equation} \label{eq:entfid_M}
F(M)=1-\frac{d^2-1}{4M}+\cl{O}(M^{-\frac{3}{2}+\varepsilon}).
\end{equation}
As a corollary of the result, the teleportation fidelity $\bb{F}_{\tn{T}}(M)$ of the protocol is given by
\begin{equation} \label{eq:telfid_M}
\bb{F}_{\tn{T}}(M)=1-\frac{d(d-1)}{4M}+\cl{O}(M^{-\frac{3}{2}+\varepsilon}),
\end{equation}
since [see Eq.~(\ref{eq:rel})]
\begin{equation*}
\bb{F}_{\tn{T}}(M)=\frac{dF(M)+1}{d+1}.
\end{equation*}

\section{Our generalized Port-based Teleportation} \label{main}
In this section, we take into account copies of an arbitrary mixed state for a generalized PBT
instead of copies of a maximally entangled state in the original PBT (see Fig.~\ref{Fig1}). 
This idea was first introduced in Refs.~\cite{PBP19,BPLP20}, but the entanglement fidelity and 
the teleportation fidelity for the PBT on those states are precisely investigated in this paper.

First, we analyze the entanglement fidelity of the PBT 
for copies of a depolarized state.
Note that any mixed state can always be transformed to a depolarized state with a parameter $p\in[0,1]$ 
by local operation and classical communication (LOCC)~\cite{DCLB00}. 
More precisely, 
\begin{equation}
\rho_{AB}\xlongrightarrow[\tn{LOCC}]{}\rho_{p}^{(AB)}=p\Phi_{AB}^++\frac{1-p}{d^2}\1_{AB},
\end{equation}
where $\Phi_{AB}^+$ is the maximally entangled state and $p\in[0,1]$. 
As in the scenario of the PBT, let 
\begin{equation*}
\varphi_p^{(\mathbf{AB})}:=(O_{\mathbf{A}}\otimes\1_{\mathbf{B}})\rho_{p}^{\otimes M}(O_{\mathbf{A}}^\dag\otimes\1_{\mathbf{B}})
\end{equation*}
be the $M$-port initial setting for the depolarized state
\begin{equation}
\rho_p=\left(\cl{D}_p\otimes \cl{I}_B\right)\left(\Phi_{AB}^+\right),
\label{eq:depolarized_state}
\end{equation} 
where $\cl{D}_p$ is the depolarizing channel with noise rate $1-p$, that is, 
\begin{equation}
\cl{D}_p(\rho)=(1-p)\frac{\1}{d}+p\rho.
\label{eq:depolarizing_channel}
\end{equation}
Then the depolarized teleportation channel $\tilde{\Lambda}_{\varphi_p}$ becomes
\begin{align} \label{eq:deptelfid}
\tilde{\Lambda}_{\varphi_p}(\psi_T)
&=\sum_{t=1}^M\left[\T_{\mathbf{A}\bar{B}_tT}
\sqrt{\Pi_t^{(\mathbf{A}T)}}(\varphi_p^{(\mathbf{A}\mathbf{B})}\otimes\psi_T)\sqrt{\Pi_t^{(\mathbf{A}T)}}^\dag\right]_{B_t\to B} 
\nonumber\\
&=\sum_{t=1}^M\T_{\mathbf{A}T}\Pi_t^{(\mathbf{A}T)}\left[(O_{\mathbf{A}}\otimes\1_{\mathbf{B}})
\sigma_{p}^{(t;\mathbf{A}B)}(O_{\mathbf{A}}^\dag\otimes\1_{\mathbf{B}})\otimes\psi_T\right],
\end{align}
where $\sigma_{p}^{(t;\mathbf{A}B)}$ in Eq.~(\ref{eq:deptelfid}) can be represented as follows.

\begin{lemma}\label{lem:sigmap} For any $p\in[0,1]$,
\begin{align}
\sigma_{p}^{(t;\mathbf{A}B)}=\sum_{r=0}^M&p^{M-r}(1-p)^r\Big[\begin{pmatrix}M-1 \\ r\end{pmatrix}\frac{1}{d^{M-1}}\Phi_{A_tB_t}^+\otimes\1_{\bar{A}_t} \nonumber\\
&+\begin{pmatrix}M-1 \\ r-1\end{pmatrix}\frac{1}{d^{M+1}}\1_{\mathbf{A}}\otimes\1_{\bar{B}_t}\Big]_{B_t\to B}.
\label{eq:sigmap}
\end{align}
\end{lemma}
\begin{proof}
Since 
\begin{equation}
\sigma_{p}^{(t;\mathbf{A}B)}=\left[\T_{\bar{B}_t}\left(\rho_p^{(A_1B_1)}\otimes\cdots\otimes\rho_p^{(A_MB_M)}\right)\right]_{B_t\to B},
\label{eq:lem_proof}
\end{equation}
it is tedious but straightforward to obtain the equality in Eq.~(\ref{eq:sigmap}). 
\end{proof}

\begin{theorem}
\label{thm:depentfid}
\begin{equation}
F(\tilde{\Lambda}_{\varphi_p})=pF(\tilde{\Lambda}_{\varphi})+\frac{1}{d^2}(1-p),
\end{equation}
where $\varphi$ is the state in Eq.~(\ref{eq:varphi_AB}), that is, $\varphi=\varphi_1$.
\end{theorem}
\begin{proof}
By using Lemma~\ref{lem:sigmap}, it can be shown that
\begin{align}
F(\tilde{\Lambda}_{\varphi_p})
=&\frac{1}{d^2}\sum_{t=1}^M\T\Pi_t^{(\mathbf{A}B)}\left[(O_{\mathbf{A}}\otimes\1_B)
\sigma_{p}^{(t;\mathbf{A}B)}(O_{\mathbf{A}}^\dag\otimes\1_B)\right] \nonumber\\
=&F(\tilde{\Lambda}_{\varphi})\cdot\sum_{r=0}^Mp^{M-r}(1-p)^r\begin{pmatrix}M-1 \\ r \end{pmatrix}\nonumber\\
&+\frac{1}{d^2}\sum_{r=0}^{M}p^{M-r}(1-p)^r\begin{pmatrix}M-1 \\ r-1 \end{pmatrix} \nonumber\\
=&F(\tilde{\Lambda}_{\varphi})\cdot\sum_{r=0}^{M-1}pp^{M-1-r}(1-p)^r\begin{pmatrix}M-1 \\ r \end{pmatrix}\nonumber \\
&+\frac{1}{d^2}\sum_{r=0}^{M-1}p^{M-1-r}(1-p)^{r}(1-p)\begin{pmatrix}M-1 \\ r \end{pmatrix} \nonumber\\
=&pF(\tilde{\Lambda}_{\varphi})+\frac{1}{d^2}(1-p), 
\end{align}
since
\begin{equation}
\sum_{t=1}^M\T\Pi_t^{(\mathbf{A}B)}\left[(O_{\mathbf{A}}\otimes\1_{B})
\tfrac{1}{d^{M+1}}\1_{\mathbf{A}}\otimes\1_{B}(O_{\mathbf{A}}^\dag\otimes\1_{B})\right]=1.
\end{equation}
\end{proof}
By Theorem~\ref{thm:depentfid}, we clearly obtain the following corollary.
\begin{corollary}
\label{cor:depentfid} 
\begin{align}
\bb{F}_{\tn{T}}(\tilde{\Lambda}_{\varphi_p})=&\frac{d^2pF(\tilde{\Lambda}_{\varphi})+d+1-p}{d(d+1)}\\
=& p\bb{F}_{\tn{T}}(\tilde{\Lambda}_{\varphi})+(1-p)(1+\tfrac{1}{d}).
\end{align}
\end{corollary}

We now note that the depolarized fully entangled fraction of $\rho_p$ can be obtained by
\begin{align}
f(\rho_p)&=\max_{\forall\ket{\Phi^+}}\bra{\Phi^+}\rho_p\ket{\Phi^+} \\
&=p+\frac{1-p}{d^2}=\frac{1+(d^2-1)p}{d^2},
\end{align}
and we thus have an identity $p=\frac{d^2f(\rho_p)-1}{d^2-1}$.

For $p\in[0,1]$, 
let $F^{(p)}$ and $\bb{F}^{(p)}_{\tn{T}}$ be the entanglement fidelity and the teleportation fidelity 
for the standard PBT protocol over the state $\rho_p^{\otimes M}$, respectively.   
Then by Theorem~\ref{thm:depentfid} and Corollary~\ref{cor:depentfid}, 
the $F^{(p)}$ and $\bb{F}^{(p)}_{\tn{T}}$ can be described as follows. 
\begin{theorem} \label{prop:genfids}
For sufficiently large $M\gg0$ and for any $\varepsilon>0$, we obtain that
\begin{align}
F^{(p)}(M)&=f(\rho_p)-\frac{p(d^2-1)}{4M}+\cl{O}(M^{-\frac{3}{2}+\varepsilon}), \\
\bb{F}_{\tn{T}}^{(p)}(M)&=\bb{F}_{\tn{T}}({\Lambda}_{\rho_p})-\frac{pd(d-1)}{4M}+\cl{O}(M^{-\frac{3}{2}+\varepsilon}).
\end{align}
\end{theorem}

Finally, we investigate the \emph{generalized} teleportation fidelity and the entanglement fidelity 
for copies of an arbitrary mixed state shared between Alice and Bob. 

\begin{theorem}\label{thm:genfids}
For any bipartite mixed state $\rho$ over Alice and Bob,
i.e., mixed teleportation channel $\tilde{\Lambda}_{\rho^{\otimes M}}$,  
the entanglement fidelity $F^\rho$ and the teleportation fidelity $\bb{F}_{\tn{T}}^{\rho}$
for the standard PBT on $\rho^{\otimes M}$ are shown as
\begin{align}
F^{\rho}(M)
&=f(\rho)\left(1-\frac{d^2}{4M}\right)+\frac{1}{4M}+\cl{O}(M^{-\frac{3}{2}+\varepsilon}),\nonumber\\
\bb{F}_{\tn{T}}^{\rho}(M)
&=\bb{F}_{\tn{T}}({\Lambda}_{\rho})\left(1-\frac{d^2}{4M}\right)+\frac{d}{4M}+\cl{O}(M^{-\frac{3}{2}+\varepsilon}).
\label{eq:gentelfid_M}
\end{align}
\end{theorem}

\begin{proof}
Recall that $p=\frac{d^2f(\rho_p)-1}{d^2-1}=\frac{d^2f(\rho)-1}{d^2-1}$, 
since without loss of generality, we may assume that $\rho$ can be transformed to $\rho_p$ satisfying $f(\rho)= f(\rho_p)$ under LOCC. 
Thus from Theorem~\ref{prop:genfids}, we have
\begin{align*}
F^{\rho}(M)&=f(\rho)-\tfrac{d^2f(\rho)-1}{d^2-1}\cdot\tfrac{d^2-1}{4M}+\cl{O}(M^{-\frac{3}{2}+\varepsilon}) \\
&=f(\rho)\left(1-\tfrac{d^2}{4M}\right)+\tfrac{1}{4M}+\cl{O}(M^{-\frac{3}{2}+\varepsilon})
\end{align*}
and 
\begin{align*}
\bb{F}_{\tn{T}}^{\rho}(M)&=\frac{dF^{\rho}(M)+1}{d+1} \\
&=\bb{F}_{\tn{T}}({\Lambda}_{\rho})\left(1-\tfrac{d^2}{4M}\right)+\frac{d}{4M}+\cl{O}(M^{-\frac{3}{2}+\varepsilon}).
\end{align*}
This completes the proof.
\end{proof}

\section{Controlled Port-based teleportation and its control power} \label{CPBT}
In this section, we propose an extended concept of the PBT, 
called the CPBT, 
which combines the well-known two ideas about CT and PBT. 
Under this consideration, Charlie can be endowed with assistance ability as a control power. 

Let $\varphi_{ABC}:=\proj{\varphi}{\varphi}_{ABC}$ be a three-qudit pure state. 
For any $\alpha\in\{A,B,C\}$, 
let the maximal CPBT fidelity $\bb{F}_{\tn{CT}}^{\varphi_{ABC},\alpha}$ be the maximal teleportation fidelity   
(depending on the port $M$) of the resulting two-qudit state $\rho_{\beta\gamma}^{[i]}$ in the subsystem $\beta\gamma$,
where $\{\alpha,\beta,\gamma\}=\{A,B,C\}$, 
after the $d$-outcome measurement $\{\cl{M}_i\}$ performed on the state $\rho_{\alpha}=\T_{\beta\gamma}\varphi_{ABC}$.
In other words, it mathematically becomes  
\begin{align} \label{eq:maxtelfid}
\bb{F}_{\tn{CT}}^{\varphi_{ABC},\alpha}(M)
=&\max_{\{\cl{M}_i\}}\sum_{i=0}^{d-1}\T(\cl{M}_i\rho_\alpha)
\bb{F}_{\tn{T}}^{\rho_{\beta\gamma}^{[i]}}(M) \\
=&\max_{\{\cl{M}_i\}}\sum_{i=0}^{d-1}\T(\cl{M}_i\rho_\alpha)
\bb{F}_{\tn{T}}({\Lambda}_{\rho_{\beta\gamma}^{[i]}})\left(1-\tfrac{d^2}{4M}\right)
\nonumber\\
&+\tfrac{d}{4M}+\cl{O}(M^{-\frac{3}{2}+\varepsilon}) \nonumber\\
=&\bb{F}_{\tn{CT}}^\alpha(\varphi_{ABC})\left(1-\frac{d^2}{4M}\right) 
+\frac{d}{4M}+\cl{O}(M^{-\frac{3}{2}+\varepsilon}),
\end{align}
where the maximum is taken over all measurements $\{\cl{M}_i\}$ on the subsystem $\alpha$, 
and 
$\bb{F}_{\tn{CT}}^\alpha$ is called the maximal CT fidelity,
which is defined as 
the maximal teleportation fidelity of the resulting two-qudit state in the subsystem $\beta\gamma$
after the measurement on the system $\alpha$ as in Refs.~\cite{LJK05,LJK07,JKL16},
that is,  
\begin{equation}
\bb{F}_{\tn{CT}}^\alpha(\varphi_{ABC})
=\max_{\{\cl{M}_i\}}\sum_{i=0}^{d-1}\T(\cl{M}_i\rho_\alpha)
\bb{F}_{\tn{T}}({\Lambda}_{\rho_{\beta\gamma}^{[i]}}).
\label{eq:F_CT}
\end{equation}
We notice that since $\bb{F}_{\tn{CT}}^{\varphi_{ABC},\alpha}(M)\ge\bb{F}_{\tn{T}}^{\rho_{\beta\gamma}}(M)$, 
one can naturally define a concept of the control power for the CPBT on a given tripartite quantum state
as the difference between the maximal CPBT fidelity and 
the teleportation fidelity without control.

For each $\alpha\in\{A,B,C\}$ and a sufficiently large port number $M>0$, 
let the control power $\bb{P}_{\tn{CT}}^{\varphi_{ABC},\alpha}(M)$ of the state $\varphi_{ABC}$ be defined as
\begin{equation}
\bb{P}_{\tn{CT}}^{\varphi_{ABC},\alpha}(M)=\bb{F}_{\tn{CT}}^{\varphi_{ABC},\alpha}(M)-\bb{F}_{\tn{T}}^{\rho_{\beta\gamma}}(M),
\label{eq:CP}
\end{equation}
and the minimal control power $\bb{P}_{\tn{CT}}(M)$ of the state $\varphi_{ABC}$ be defined as
\begin{equation}
\bb{P}_{\tn{CT}}(M)=\min_{\forall \alpha\in\{A,B,C\}}\bb{P}_{\tn{CT}}^{\varphi_{ABC},\alpha}(M).
\label{eq:MCP}
\end{equation}
Then the following corollary can clearly be obtained.

\begin{corollary}\label{prop:pbcp}
For each $\alpha\in\{A,B,C\}$, the control power of a state $\varphi_{ABC}$ can simply be expressed as 
\begin{align}
\bb{P}_{\tn{CT}}^{\varphi_{ABC},\alpha}(M)
=\bb{P}_{\tn{CT}}^\alpha(\varphi_{ABC})\left(1-\tfrac{d^2}{4M}\right)+\cl{O}(M^{-\frac{3}{2}+\varepsilon}),
\end{align}
where $\bb{P}_{\tn{CT}}^\alpha(\varphi_{ABC})$ is the control power for the CT proposed in Ref.~\cite{JKL16}, 
that is, 
\begin{equation*}
\bb{P}_{\tn{CT}}^\alpha(\varphi_{ABC})=\bb{F}_{\tn{CT}}^\alpha(\varphi_{ABC})-\bb{F}_{\tn{T}}({\Lambda}_{\rho_{\beta\gamma}}).
\end{equation*}
\end{corollary}

Hence by exploiting Corollary~\ref{prop:pbcp} and the results in Ref.~\cite{JKL16}
we can readily compute the control powers for three-qubit pure states, 
especially, the extended Greenberger-Horne-Zeilinger (GHZ) state, $\ket{\varphi_{\tn{GHZ}}}=a\ket{000}+b\ket{111}$, 
and the $W$-class state, $\ket{\varphi_W}=w_0\ket{000}+w_1\ket{100}+w_2\ket{101}+w_3\ket{110}$ 
with the coefficients $w_i\ge0$ such that $\sum_i w_i^2=1$~\cite{AA+00,ABLS01}.

Since it is shown~\cite{JKL16} that for any $\alpha\in\{A,B,C\}$, 
the control powers for the CT on the extended GHZ state and the $W$-class state are
\begin{align*}
&\bb{P}_{\tn{CT}}^\alpha(\varphi_{\tn{GHZ}})=\frac{2|a||b|}{3},\\
&\bb{P}_{\tn{CT}}^\alpha(\varphi_{W})=\frac{1}{6}(2w_\beta w_\gamma+1-\sqrt{W_\alpha}),
\end{align*}
where 
\begin{equation*}
W_\alpha=\max\left\{(w_1^2+(\mp w_\beta\pm w_\gamma+\alpha_l)^2)(w_1^2+(w_j+w_k\pm w_l)^2)\right\},
\end{equation*}
the control powers for the CPBT on the extended GHZ state and the $W$-class state become
\begin{equation}
\bb{P}_{\tn{CT}}(M)=\bb{P}_{\tn{CT}}^{\varphi_{\tn{GHZ}},\alpha}(M)
=\frac{2|a||b|}{3}\left(1-\frac{1}{M}\right)+\cl{O}(M^{-\frac{3}{2}+\varepsilon})
\end{equation}
and 
\begin{equation}
\bb{P}_{\tn{CT}}^{\varphi_W,\alpha}(M)
=\frac{1}{6}(2w_\beta w_\gamma+1-\sqrt{W_\alpha})
\left(1-\frac{1}{M}\right)+\cl{O}(M^{-\frac{3}{2}+\varepsilon}),
\end{equation}
respectively.

\section{Conclusions} \label{conclusion}
In this paper, 
we have introduced quantities for controlled teleportation capability,
as a variant of the teleportation capability of original quantum teleportation (or the PBT and the CT), 
and have analyzed its control powers in terms of the port number $M$ 
representing how faithfully the PBT on a given tripartite state can be performed. 
Here, we have made use of a generalization technique 
to employ the PBT channel on copies of an arbitrary bipartite mixed state 
instead of the standard one on copies of a pure bipartite maximally entangled state.
Furthermore, we have found explicit formulas and relations for the teleportation fidelity and the entanglement fidelity 
as well as the maximal teleportation fidelity, and thus we have derived the control powers for the CPBT on tripartite quantum states. 

There are still some intriguing open questions on the PBT itself and beyond. 
First, we can imagine another variants of the PBT protocols or their communication capabilities in quantum communication, 
for example, a controlled dense-coding~\cite{OKJJ17} and a remote state-preparation scheme~\cite{NK08}. 
Since the performance of the PBT has been studied in a probabilistic scenario as well as in a deterministic one, 
we can also investigate how the controlled version behaves in a probabilistic scenario. 
Those kind of researches may extend our knowledge of port-based quantum communications.

\section{Acknowledgments}
This research was supported by the Basic Science Research Program 
through the National Research Foundation of Korea funded by the Ministry of Science and ICT (Grant No. NRF-2019R1A2C1006337) and the Ministry of Science and ICT, Korea, under the Information Technology Research Center support program (Grant No. IITP-2019-2018-0-01402) 
supervised by the Institute for Information and Communications Technology Promotion.
K.J. acknowledges support from Basic Science Research Program through the National Research Foundation of Korea, a grant funded by the the Ministry of Education (NRF-2018R1D1A1B07047512) and the Ministry of Science and ICT (NRF-2017R1E1A1A03070510). J.K. was supported in part by KIAS Advanced Research Program CG014604. S.L. acknowledges support from Research Leave Program of Kyung Hee University in 2018.

%

\end{document}